\documentclass[10pt]{article} % use larger type; default would be 10pt

\usepackage{amsmath}
\usepackage{amsfonts}
\usepackage[utf8]{inputenc}

\usepackage{amsmath,amscd}
\usepackage[english]{babel}
\usepackage{amsthm}
\usepackage{fancyhdr}
\usepackage{latexsym}
\usepackage{amssymb}
\usepackage{color}
\usepackage{makeidx}
\usepackage{graphicx}
\usepackage{amssymb}
\usepackage{epstopdf}
\usepackage[mathscr]{eucal}
\usepackage{stackrel}
\usepackage{enumerate}
\usepackage{setspace}
\usepackage{comment}
\usepackage{mathrsfs}

\newtheorem{lemma}{Lemma}[section]

\newenvironment{remark}[1][Remark]{\begin{trivlist}
\item[\hskip \labelsep {\bfseries #1}]}{\end{trivlist}}

\newcommand{\Imm}{\textnormal{Im}}

\newcommand{\GL}{\textnormal{GL}}

\newcommand{\vF}{\mathbb{F}}
\newcommand{\vZ}{\mathbb{Z}}

\def\vF{\mathbb{F}}

\begin{document}

%\title{A Geometric Description of Alternant Codes \thanks{The authors
%    were supported in part by Swiss National Science Foundation Grant
%    no. 138080.}} \author{Kyle Marshall \and Giacomo Micheli \and
%  Joachim Rosenthal} \institute{The author are at \at
%  Institute of Mathematics, University of Zurich,\\
%  Winterthurerstr 190, 8057 Zurich, Switzerland } \date{\today}

%\maketitle

\sloppy

%% Paper Title
%% You can use linebreaks \\ within to get better formatting as
%% desired. 
\title{Cryptanalysis of a non-commutative key exchange protocol}
%% Author names and affiliations:
%%
%% Avoiding spaces at the end of the author lines is not a problem with
%% conference papers because we don't use \thanks or \IEEEmembership.
%%
%% For several authors with only one affiliation:
%%
% \author{
%   \IEEEauthorblockN{Hui-Ting Chang and Stefan M.~Moser}
%   \IEEEauthorblockA{Department of Electrical and Computer Engineering\\
%     National Chiao Tung University (NCTU)\\
%     Hsinchu, Taiwan\\
%     Email: \{email-of-hui-ting,email-of-stefan\}@ieee.org} 
% }
%%
%% For up to three affiliations:
%%
\author{
  Giacomo Micheli\\
    Institute für Mathematik\\
    Universität Zürich\\
    Switzerland\\
    Email: giacomo.micheli@math.uzh.ch}

\maketitle

\begin{abstract}
In the papers by Alvarez et al. and Pathak and Sanghi a non-commutative based public key exchange is described.
A similiar version of it has also been patented (US7184551).
In this paper we present a polynomial time attack that breaks the variants of the protocol
presented in the two papers. Moreover we show that breaking the patented cryptosystem US7184551
can be easily reduced to factoring. We also give some examples to show how
efficiently the attack works.
\end{abstract}

%\begin{keywords}
%Cryptography, Cryptanalysis, Security, Public key, DLP, Finite fields, Diffie-Hellmann, NonCommutative Algebra
%Chinese Remainder Theorem, Cayley-Hamilton Theorem.
%\end{keywords}

\section{Introduction}

We first describe the noncommutative key exchange presented in \cite{RA2007}
and \cite{Pathak}.
Consider the group $G=\GL(n,\vF_q)$ of invertible matrices over the finite field $\vF_q$ and $M_1,M_2\in G$.
Let the triple $(G,M_1,M_2)$ be public, with $M_1,M_2\in\GL(n,\vF_q)$. 
Let $M_1,M_2$ be elements in $G$ such that $M_1 M_2 \neq M_2 M_1$.
Let the public key be $(G,M_1,M_2)$.
\begin{itemize}
\item{Alice chooses $(a_1,a_2)\in\vZ^2$ and sends $C_1=M_1^{a_1} M_2^{a_2}$ to Bob}
\item{Bob chooses $(b_1,b_2)\in \vZ^2$ and sends $C_2=M_1^{b_1}C_1 M_2^{b_2}$ to Alice}
\item{Alice computes $K=M_1^{-a_1}C_2 M_2^{-a_2}$}
\end{itemize}

As a result Alice and Bob can compute the secret key $K=M_1^{b_1}M_2^{b_2}$.
The purpose of this paper is to show that $K$ can be computed in $O(n^3)$ field operations from $C_1$ and $C_2$.

\section{Preliminaries}

\subsection{Cayley-Hamilton Theorem}
Let $f_M$ be the characteristic polynomial of $M$. Then $f_M(M)=0$. As a result every power of $M$ can be written
in terms of a linear combination of $\{1,M,\dots,M^{n-1}\}$ where $n$ is the order of the matrix. A nice proof of this is presented in \cite[p.21]{ATMC}

\subsection{Solving homogeneous ``mixed'' multivariate polynomial equations of degree $2$}
In general solving multivariate polynomial equations is NP-complete \cite{EASOE}.
In situations where the number of variables is much smaller than the number of equations there exists a polynomial time algorithm. To be precise:
 \cite{EASOE} proposes
an ''expected'' polynomial time algorithm  to solve 
overdefined polynomial 
equations when the number of unknowns $k$ and the number of equations $m$ satisfy the inequality $m\geq \varepsilon k^2$
where $\varepsilon \in (0,{1 \over 2}]$. The expected running time is approximately of 
$k^{O({1 \over \sqrt{\varepsilon}})}$.
The aim of the paper is to use a specialized version of this algorithm that can be proved to lead to a polynomial time algorithm that breaks the protocols
described in \cite{RA2007} and \cite{Pathak} in the generic condition. For all other cases we refer to the euristic result in \cite{EASOE} that, for cryptanalytic purposes,
will be enough.

\subsection{Commutative rings of matrices}
Let $T$ be an invertible matrix and define $\vF_q[T]$ to be the $\vF_q$-algebra generated by $T$.
In other words it will be the image of the evaluation map
\[\psi: \vF_q[x]\longrightarrow M_{n \times n}(\vF_q)\]
\[\psi(p(x)):=p(T).\]
By Cayley Hamilton theorem it follows that $\vF_q[T]$ is a finite dimensional algebra over $\vF_q$.
The following short lemma will be useful for our purposes:
\begin{lemma}\label{RR}
Let $p(T)\in \vF_q[T] \cap \GL_n(\vF_q)=:\left(\vF_q[T]\right)^*$ and 
\[C(T):=\{L\in \GL_n(\vF_q)\: | \: LT=TL \}.\] Then $p(T)^{-1}\in C(T)$.
\end{lemma}
\begin{proof}
\[p(T)^{-1} T= (T^{-1} p(T))^{-1}=(p(T)T^{-1})^{-1}=T p(T)^{-1}\]
\end{proof}

\section{Performing the attack}
The attack described in the following sections makes use of the elementary tools mentioned above and this is intended
to show the structural vulnerabilities of the system.
Suppose Eve is observing the key exchange, she is then able to get the following information:
$\{M_1,M_2,M_1^{a_1} M_2^{a_2}, M_1^{a_1+b_1} M_2^{a_2+b_2}\}$.
Eve first observes $M_1^{a_1}\in\vF_q[M_1]$ and $M_2^{a_2}\in\vF_q[M_2]$ and that $M_1^{a_1} M_2^{a_2}$ is in
the image of the application
\[\varphi: \vF_q[M_1]\times \vF_q[M_2] \longrightarrow M_{n\times n}(\vF_q)\]
\[\varphi(h_1(M_1),h_2(M_2))=h_1(M_1)h_2(M_2).\]
Then, we are able to write $M_1^{a_1} M_2^{a_2}=p(M_1)q(M_2)$, for some $p(x),q(x)\in \vF_q[x]$ with
\[p(M_1)=\sum_{i=0}^{n-1} x_i M_1^i \quad q(M_2)=\sum_{j=0}^{n-1}y_j M_2^j\]
and then
\begin{equation}\label{a1a2}
M_1^{a_1} M_2^{a_2}=\sum_{i=0,\,j=0}^{n-1} x_i y_j M_1^i M_2^j
\end{equation}
where $x_i,y_j$ are indeterminates.
Observe that the system is solvable in polynomial time with $m=n^2$, $k=2n$ and $\varepsilon= {1 \over 4}$ and expected running time $O(n^{2})$.
We pick now \emph{any} solution and write down $p(M_1)$ and $q(M_2)$.

\begin{remark}
The system given by Equation \ref{a1a2} is easy to solve, again
even without the knowledge of the algorithm presented in \cite{EASOE}
since they consist of $n^2$ homogeneous equations of degree $2$ in $2n$ unknowns
where we can perform a Gaussian elimination-like computation on the variables $u_{i,j}:=x_i y_j$.
We will show those equations with an explicit example in the next subsection.
It is also elementary to observe that when the $n^2$ by $n^2$ matrix of the linear system 
\[\sum_{i,j}u_{i,j}M_1^i M_2^j=M_1^{a_1} M_2^{a_2}\]
is invertible, the attack can be proven to be polynomial by the observation
that the $u_{i,j}$ are unique and the system $x_i y_j= u_{i,j}$ admits a solution by construction (that can be found just by substitutions).
In particular this happens when we have the non degenerate case, in the sense that the $k$-vector space generated by the $M^iN^j$ is the whole matrix ring.

\end{remark}

\begin{remark}
$\,$
\begin{itemize}
\item{We have at least one solution by the observation $M_1^{a_1} M_2^{a_2}\in \Imm{\varphi}$.}
\item{We are not claiming $p(M_1)=M_1^{a_1}$ and $q(M_2)=M_2^{a_2}$. It has been observed in \cite{MMR} 
that in order to break the protocol it is enough to find $U\in\vF_q[M_1]$ and $V\in\vF_q[M_2]$ such that $U=M_1^{a_1}$ and $V=M_2^{a_2}$ }
\item{Since $M_1,M_2$ are invertible, also $p(M_1)$ and $q(M_2)$ are.}
\item{$M_1^{a_1} M_2^{a_2}=p(M_1)q(M_2)$ implies \[(p(M_1)^{-1}M_1^{a_1})( M_2^{a_2}q(M_2)^{-1})=1\]}
\end{itemize}
\end{remark}

Eve gets the key thanks to the computation
\[p(M_1)^{-1}M_1^{a_1+b_1} M_2^{a_2+b_2}q(M_2)^{-1}=\]
\[p(M_1)^{-1}M_1^{b_1}M_1^{a_1}M_2^{a_2} M_2^{b_2}q(M_2)^{-1}=\]
\[M_1^{b_1}(p(M_1)^{-1}M_1^{a_1})(M_2^{a_2} q(M_2)^{-1})M_2^{b_2}=\]
\[M_1^{b_1}(p(M_1)^{-1}M_1^{a_1}M_2^{a_2} q(M_2)^{-1})M_2^{b_2}=\]
\[M_1^{b_1}\cdot 1\cdot M_2^{b_2}=M_1^{b_1}M_2^{b_2}=K\]
where the second equality is due to lemma \ref{RR} and the very last one by
the solution of the system \ref{a1a2}.
\subsection{Example}
Let $p=569$ and $G=\GL_2(\vF_{569})$,
\[M_1=\left(
\begin{array}{cc}
 12 & 34 \\
 11 & 99 \\
\end{array}
\right)\]
and
\[M_2=\left(
\begin{array}{cc}
 172 & 94 \\
 91 & 125 \\
\end{array}
\right).\]
Alice chooses $(a_1,a_2)=(449,41)$. Bob chooses $(b_1,b_2)=(509,131)$.
Alice computes 
\[C_1=M_1^{a_1}M_2^{a_2}=\left(
\begin{array}{cc}
 502 & 108 \\
 3 & 322 \\
\end{array}
\right)\]
and sends that to Bob.
Bob computes the secret key
\[K=M_1^{b_1}M_2^{b_2}=
\left(
\begin{array}{cc}
 273 & 85 \\
 436 & 278 \\
\end{array}
\right)
\]
and
\[M_1^{b_1}CM_2^{b_2}=\left(
\begin{array}{cc}
 501 & 343 \\
 200 & 170 \\
\end{array}
\right)\]
and sends this last value to Alice.
Alice now decrypts as $M_1^{-a_1}M_1^{b_1}CM_2^{b_2}M_2^{-a_2}=K$.
Eve writes the system
\[(x_0 Id + x_1M1)(y_0 Id + y_1M2)=C\]
consisting of the equations
\[
\begin{cases}
 {x_0} {y_0}+172 {x_0} {y_1}+12 {x_1}
   {y_0}+37 {x_1} {y_1} = 502
   \\
   94 {x_0} {y_1}+34
   {x_1} {y_0}+257 {x_1} {y_1}=108 \\
 91 {x_0} {y_1}+11 {x_1} {y_0}+90 {x_1}
   {y_1}= 3
   \\
   {x_0} {y_0}+125 {x_0} {y_1}+99 {x_1} {y_0}+322 {x_1} {y_1}=322
\end{cases}
\]
Eve can easily find a solution of this system by a relinearization of the kind
$x_i y_j=:u_{i,j}$; e.g. $(x_0,x_1,y_0,y_1)=(1,166,244,168)$.
She now performs the computation
\[(x_0 Id + x_1M1)^{-1}M_1^{a_1+b_1}M_2^{a_2+b_2}  (y_0 Id + y_1M_2)^{-1}=\]
\[
\left(
\begin{array}{cc}
 273 & 85 \\
 436 & 278 \\
\end{array}
\right)=K\]

\section{Cryptanalysis of the public key patented variant of the protocol}

\subsection{Brief description}

In this section we cryptanalise the patent $US7184551$. 
Observe that the patented protocol is roughly $k$ times computationally more expensive than RSA, where $k$ is the order of the matrices we are using.

\emph{Public key}
\begin{itemize}
\item{Alice chooses $A,C\in GL_k(\vZ_n)$ for $n=pq$ and $p,q$ prime numbers}
\item{$B=CAC$ and $G\in\vZ_n[C]$}
\item{Alice publishes $(A,B,G)$}
\end{itemize}
\emph{Encryption}
\begin{itemize}
\item{Bob choses $D\in \vZ_n[G]$}
\item{Bob computes $K=DBD$ and $E=DAD$}
\item{Let $M$ be the message in $GL_k(\vZ_n)$}
\item{Bob sends $(KM,E)$}
\end{itemize}
\emph{Decryption}
\begin{itemize}
\item{Alice computes $CEC=K$}
\item{Alice decrypts as $K^{-1} K M$}
\end{itemize}

\subsection{Cryptanalysis}
The idea behind this cryptanalysis is the same as in the previous sections,
we just need to make a revision of what we did before.
In this section we prove that the problem of breaking the protocol above can be reduced to factoring a modulus.
If $M\in M_{k\times k}(\vZ_n)$, let $M_p$ and $M_q$ denote
its two reductions modulo $p$ and $q$ respectively.
We reduce $G$, $E$ and $A$ modulo $p$ and write the system
\begin{equation}\label{system2}
E_p=\left(\sum_{j=0}^{k-1}x_j G_p^j\right)A_p\left(\sum_{i=0}^{k-1}x_i G_p^j\right)
\end{equation}
in $k$ unknowns and $k^2$ homogeneous degree $2$ equations over $\vF_p^k$.
We can assure at least one solution by the construction of $E_p=D_pA_pD_p$, since $D_p$
can be written in terms of low powers of $G$.
We apply again the algorithm presented in \cite{EASOE} getting one solution for the system
in polynomial time with $\varepsilon=1/2$.
This solution identifies a matrix $D'\in M_{k\times k}(\vF_p)$ such that
\[D' A_p D'=D_p A_p D_p=E_p \quad \mod p.\]
Observe that we get the partial secret key $K_p=K \mod p$ by multiplying 
$B_p$ on both sides by $D'$
\[D'B_pD'=D'C_pA_pC_pD'=\]
\[C_pD'A_pD'C_p=C_pD_pA_pD_pC_p=K_p \mod p.\]
We perform the same procedure modulo $q$
getting $D''M_{k\times k}(\vF_q)$ such that $D''B_qD''=K_q$.
Since we have the computable isomorphism of rings
\[\psi: M_{k\times k}(\vZ_n)\longrightarrow M_{k\times k}(\vF_p)\oplus 
M_{k\times k}(\vF_q)\]
given by the Chinese Remainder Theorem
we are able to recover the secret key $K$ just by taking the preimage 
of the pair $(K_p,K_q)$ through $\psi$.
Note that $\psi^{-1}(K_p,K_q)$ is exactly $K$
by observing that $K_p$ and $K_q$ are necessarily the reductions 
of $K$ modulo $p$ and $q$ (by the homomorphism properties of $\psi$)
and then that $K=\psi^{-1}(K_p,K_q)$ since $\psi$ is a bijection.

\begin{remark}
Observe that the equations in (\ref{system2}) are even easier then the ones in (\ref{a1a2})
since they have the same structure but half of the unknowns. 
What is again important to observe is that Cayley Hamilton theorem always assures us
a solution.
In the next subsection we give an example to show how they look like.
\end{remark}

\subsection{Example}

Let $n=6133=541\cdot 113$ and Alice's public key
constructed as follow: let $C$ be
\[\left(
\begin{array}{cc}
 243 & 112 \\
 234 & 233 \\
\end{array}
\right)\]
and
\[A=\left(
\begin{array}{cc}
 121 & 231 \\
 144 & 242 \\
\end{array}
\right)\]
then 
\[
B=CAC=
B=\left(
\begin{array}{cc}
 36124 & 40493 \\
 39554 & 16490 \\
\end{array}
\right)
\]
Then we choose $G=14 \cdot 1+ 3374 \cdot C=$
\[
\left(
\begin{array}{cc}
 25167 & 11090 \\
 55920 & 52560 \\
\end{array}
\right).\]
The public key will be $(A,B,G)$.
Bob secret key is constructed as follow:
Let
\[D=34125\cdot 1 + 7123 G=\left(
\begin{array}{cc}
 56710 & 10234 \\
 36665 & 40513 \\
\end{array}
\right)\] 
\[K=\left(
\begin{array}{cc}
 20609 & 51651 \\
 14785 & 1448 \\
\end{array}
\right)\]
and
\[E=DAD=
\left(
\begin{array}{cc}
 57174 & 14133 \\
 7237 & 20711 \\
\end{array}
\right).
\]
Let $M$ be any message, then Bob sends
\[(KM,E)\]
Eve attacks the system as follows:
she reduces $E$ modulo $541$ getting
\[E_{541}=\left(
\begin{array}{cc}
 369 & 67 \\
 204 & 153 \\
\end{array}
\right).
\]
She has now to solve the system
\[(x 1+ y G_{541})A_{541}(x 1+ y G_{541})=E_{541}\]
getting for example the solution
$(x_0,y_0)=(220,159)$,
so we get 
$K_{541}=(x_0 1+ y_0 G_{541}) B_{541} (x_0 1+ y_0 G_{541})$
and then
\[K_{541}=\left(
\begin{array}{cc}
 51 & 256 \\
 178 & 366 \\
\end{array}
\right)\]
Analogously one gets
$K_{113}=(x_0 1+ y_0 G_{113}) B_{113} (x_0 1+ y_0 G_{113})$
where $(x_0,y_0)=(55,49)$, getting
\[
\left(
\begin{array}{cc}
 43 & 10 \\
 95 & 92 \\
\end{array}
\right)
\]
By the Chinese Remainder Theorem we get
\[K=\left(
\begin{array}{cc}
 20609 & 51651 \\
 14785 & 1448 \\
\end{array}
\right)\]

\nocite{MMR}

\section{Conclusion}
We have presented a polynomial time attack to the noncommutative protocol proposed in  \cite{RA2007} and \cite{Pathak}.
Moreover we have shown the weakness of 
such a protocol over any subgroup of matrices over
any finite field.
We have also made the attack work on the protocol presented in the patent \cite{PatentG}
showing that breaking the cryptosystem can be reduced to factoring. 
It would be very interesting to find analogous noncommutative schemes that are resistant to the attack we presented.
In particular the key point is that the vector space structure of matrix rings over fields is a major weakness of such kind of protocols.
We would like to thank Gerard Maze and Davide Schipani for their very helpful ideas and suggestions.
\bibliography{biblio}
\bibliographystyle{plain}

\end{document}